\documentclass[final,3p,times,twocolumn]{elsarticle}

\usepackage{amsmath}
\usepackage{amsfonts}
\usepackage{amssymb}
\usepackage{color}
\usepackage{fixltx2e}
 \usepackage{graphicx}
 
\usepackage{amsthm}

\newtheorem{theorem}{Theorem}[section]
\newtheorem{lemma}[theorem]{Lemma}

\newtheorem{example}[theorem]{Example}
\newtheorem{definition}[theorem]{Definition}
\newtheorem{remark}[theorem]{Remark}

\newcommand{\diag} {\mbox{\rm diag}\,}

\newcommand{\crr}{\color{black}}
\newcommand{\crb}{\color{black}}
\newcommand{\mvec} {\mbox{\rm vec}\,}

\journal{}

\begin{document}

\begin{frontmatter}



\title{On the Reachability   of Networked Systems}

\author[label1]{Mohsen Zamani}
\author[label2]{Brett Ninness}
\author[label3]{Daniel Quevedo}
\address[label1]{School of Electrical Engineering and Computer Science, The University of Newcastle, Callaghan, NSW 2308, Australia.\\ (e-mail: mohsen.zamani@newcastle.edu.au).}
\address[label2] { School of Electrical Engineering and Computer Science, The University of Newcastle, Callaghan, NSW 2308, Australia.\\ (e-mail: brett.ninness@newcastle.edu.au). }
\address[label3]{ Department of Electrical Engineering (EIM-E), University of Paderborn,	33098 Paderborn, Germany.\\ (e-mail: dquevedo@ieee.org).}


\address{}

\begin{abstract}
In this paper, we   study networks of discrete-time linear time-invariant subsystems. Our focus is on situations where  subsystems are connected to each other through a time-invariant topology {\crb 
 and where there exists a base-station whose aim is to control the subsystems into any desired destinations. However, the base-station can only communicate with some of the subsystems that we refer to as \textit{leaders}. There  are no direct  links between the base-station and the rest of subsystems, known as \textit{followers}, as they  are only able to liaise among themselves and with some of the leaders.   

 The current paper formulates this  framework  as  the well-known reachability problem for linear systems. Then to address this problem,} we introduce notions of \textit{leader-reachability} and \textit{base-reachability}. We present  algebraic conditions under which these notions hold. It turns out that if subsystems are represented by   minimal state space representations, then base-reachability  always holds. Hence, we focus on leader-reachability and investigate the corresponding conditions in  detail. We further demonstrate that when the networked system parameters i.e. subsystems' parameters and interconnection matrices,  assume generic values then the whole network is both leader-reachable and base-reachable.

\end{abstract}


\end{frontmatter}


\section{Introduction}\label{sec:intro}

Recent developments of enabling technologies such as communication systems, cheap computation equipment and sensor platforms have given great impetus to the creation of networked systems.  Due to their large application in different branches of science and technology, these systems have attracted significant attention worldwide and researchers have studied networked systems from different perspectives (see e.g.  \cite{sinopoli2003distributed},
\cite{olfati2002distributed}, \cite{tanner2003stable}, \cite{Olfati2007}, \cite{dankers2014system}, \cite{hespanha2007survey}, \cite{zamaniautomatic2014}).

In this paper, we consider   networks consisting of finite-dimensional linear time-invariant subsystems. We suppose that each subsystem in the network has   discrete-time  dynamics and  the interconnection topology among subsystems is time-invariant. In the framework under study, there exists a base-station that can only send command signals to some of the subsystems with superior capabilities, known as \textit{leaders}. The remainder of the subsystems referred to as \textit{followers} can only accept input signals from some of the leaders and followers. 

Here, we address a fundamental issue associated with the above framework namely the reachability.  The concept of reachability is well-understood in the systems and control literature \cite{kailath}. We adopt this concept to address the following question.

\textit{
Under which conditions can the state of followers  reach any desired values using the commands generated from the base-station?}

We tackle this question by providing a   mathematical model for the networked system under study. We introduce the notions of \textit{base-reachability} and \textit{leader-reachability}.  Then we show that systems networked according to the model considered here  are generically both base-reachable and leader-reachable. This means that when the parameters of the network i.e. parameter matrices of each subsystem as well as the interconnection topology, assume generic values, these properties hold. We also investigate some topologies that give rise to  state matrices with \textit {symmetric} or \textit {circulant} structures. 

The problem studied in this paper has some connections with  the existing literature concerned with controllability of multi-agent systems. There exists a body of works in this area and among many, interested readers can refer to \cite{tanner2004}, \cite{ji2008controllability}, \cite{ji2007graph}, \cite{rahmani2009controllability}, \cite{liu2011controllability}, \cite{zamani2009structural}, \cite{ji2008controllability}, \cite{martini2010controllability} and references listed therein. These references studied the controllability problem for a group of single integrators connected through the nearest neighbourhood law.  We comment on some of the works along this line in the next paragraph. 

 The controllability problem of multi-agent systems was proposed in  \cite{tanner2004} and  the author formulated this problem  as the controllability problem of   linear systems, whose state matrices are induced from the graph \textit{Laplacian} matrix. Necessary and sufficient algebraic conditions on the state matrices were given based on linear system tools. Under the same setup, a sufficient condition was derived in \cite{Ji2006} where it was  shown that the system is controllable if the null space of the leader set is a subset of the null space of the follower set. In \cite{ji2007graph}, it was shown that a necessary and sufficient condition for controllability is not sharing any common eigenvalues between the Laplacian matrix of the follower set and the Laplacian matrix of the whole topology. However, it remains elusive on what exactly the graphical meaning of these rank conditions related to the Laplacian matrix is. This motivates several research activities on illuminating the controllability of multi-agent systems from a graph theoretic point of view. For example, a notion of anchored systems was introduced in \cite{rahmani2006}, and it was shown that symmetry with respect to the anchored vertices makes the system uncontrollable. In \cite{ji2008graph}, the authors characterized some necessary conditions for the controllability problem based on a new concept called leader-follower connectedness. While \cite{ji2008graph} was focused on the case of fixed topology, the corresponding controllability problem under switching topologies was investigated in \cite{ji2008controllability}, which employed some recent achievements in the switched systems literature. Later,  the authors of \cite{zamani2009structural} assumed the graph to be weighted with freely chosen entries. Under this setup, they proposed the notion of  structural controllability for multi-agent systems. It turned out that this controllability notion, solely depends on the topology of the communication scheme; the multi-agent system is controllable if and only if the graph is connected. This result is later extended in \cite{partovi} to the case where  the dynamics of each subsystem are expressed by high order integrators rather than a single integrator.  The authors of \cite{Ji2009} examined the connection between the controllability of networks comprising   single integrator subsystems and those consisting of subsystems with high order integrators.  

The current paper has several contributions. Firstly, in contrast to the   works described above, we relax the limitation imposed on subsystems dynamics by allowing subsystems to be  general  discrete-time  linear time-invariant (DLTI) state space systems. Secondly, in most of the literature the followers are connected to one another by the nearest neighbourhood law. We relax  this constraint here as well. Thirdly, as opposed  to existing literature, we  explicity examine the role of the base-station and its connections to the leaders.

The structure of this paper is as follows. In the next section, we formulate the problem under study.  The main results of the paper are introduced in Section \ref{sec:main}. Finally, Section \ref{sec:con} provides the concluding remarks and comments about future research directions. 

\section{Problem Formulation}

\begin{figure}[!t]
	\begin{center}
		\includegraphics[width=5cm,height=6cm]{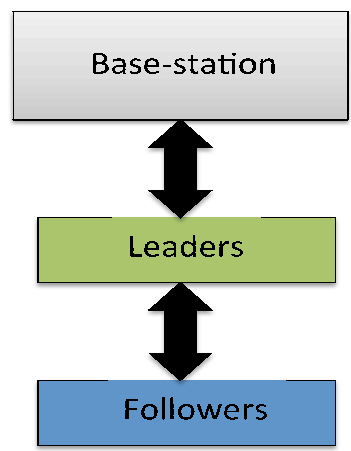}
		\caption{\emph{The connection structure between the base-station, leaders and follower}} \label{fig:structure}
	\end{center}
\end{figure}

We assume that there exist $N$ linear subsystems which are connected together through  linear coupling rules.  Suppose that there exist $N_l$ subsystems with higher levels of computing and communicating powers that we refer to as \textit{leaders}. The rest of the subsystems are called \text{followers} denoted by $N_f$. It is natural to assume that the number of leaders is strictly less than the number of followers i.e. {\crr $N_f>N_l$.} {\crb The framework studied in this paper is  depicted in Fig. \ref{fig:structure}.} 

Without loss of generality, we assume that the first $N_f$ subsystems are followers and the remaining  $N-N_f$ subsystems act as leaders.

Suppose that the linear state space  dynamics of the followers are expressed by a set of difference equations as 
\begin{equation}
\begin{split}
\label{sys1}
x^i_{t+1} &=A_ix^i_{t}+B_iv^i_{t},\\  
w^i_{t}   &=C_ix^i_{t},\; i=1,\dots,N_f.
\end{split}
\end{equation}
where $x^i_t \in \mathbb R^{n_i}$, $v^i_t \in \mathbb R ^{m_i}$, $w^i_t \in \mathbb R^{p_i}$. We suppose that all $N$ subsystems    are reachable and observable. The control command $v^i_t$ is constructed based on the following law

\begin{equation} \label{eq:followu}
v^i_{t}=\sum_{j=1}^NL_{ij}w^j_{t}. 
\end{equation}
{\crr
\begin{remark}	
Note that the control law (\ref{eq:followu}) allows consideration of both centralized and  distributed control schemes. If the control law (\ref{eq:followu}) is implemented locally, then  the control gains  $L_{ij}$ corresponding to those subsystems which are not neighbors of $i$-th subsystem are assumed to be zero. This ensures that the summation $\sum_{j=1}^N {L_{ij} {w^j_t}}$ simplifies into a summation over the neighbor set of $i$-th subsystem. Hence, the  control law (\ref{eq:followu}) represents the topology of the network  i.e. the matrices $L_{ij}$ represent which components of the state vector associated with the $j$-th subsystem are available to the local controller corresponding to the $i$-th subsystem. Thus, one can readily verify that the consensus law \cite{fax2004} can be regarded as a  special case of the control strategy \eqref{eq:followu}. 
\end{remark}
}

{\crr
Let us also define the linear  dynamics of each leader as 
 
	\begin{equation}
	\begin{split}
	\label{sysleader}
	x^i_{t+1} &=A_ix^i_{t}+B_iu_{t},\\
	w^i_{t}   &=C_ix^i_{t},\:\:\: N_{f+1},\ldots, N,
	\end{split}
	\end{equation}
	where $u_t \in \mathbb R^{m}$ is the control command generated from the base-station.

For our subsequent analysis it is convenient to define 
\begin{equation}\label{eq:nodematricesfollower}
\begin{split}
\overline  A_f  &:=\diag
(A_1, \ldots, A_{N_f}),\\
\overline B_f  &:=\diag
(B_1, \ldots, B_{N_f}),\\
C_f  &:=\diag
(C_1, \ldots, C_{N_f}),\\
L&:=\begin{pmatrix}
L_{11} & \ldots&{L_{1N}}\\
\vdots&\ddots&  \vdots\\
L_{N_f1}&  \ldots&L_{N_fN}\end{pmatrix} \in \mathbb{R}^{{m_f}\times\overline{p}},\\
x^f_t&:=
\begin{pmatrix}
x^1_t\\\vdots\\x^{N_f}_t
\end{pmatrix}\in \mathbb{R}^{n_f}, \\
v_t&:=
\begin{pmatrix}
v^1_t\\\vdots\\v^{N_f}_t
\end{pmatrix}\in \mathbb{R}^{m_f},\\
w^f_t&:=
\begin{pmatrix}
w^1_t\\\vdots\\w^{N_f}_t
\end{pmatrix}\in \mathbb{R}^{p_f}.\\
\end{split}
\end{equation}
}
where $ m_f =\sum_{i=1}^{N_f}m_i$,
$ p_f =\sum_{i=1}^{N_f}p_i$, $n_f=\sum_{i=1}^{N_f}n_i$, $\overline p=\sum_{i=1}^{N}p_i$.

We split the gain matrix $L$ as
\[L=\left(L_{ff}\:L_{lf}\right),\]
where $L_{ff}$ captures the first $p_f$ columns of $L$. This matrix captures   the interconnection existing among followers only. Furthermore,  $L_{lf}$ contains those columns of $L$ that are not contained in $L_{ff}$ and thereby exhibits the relation between followers and leaders.
 
	In terms of the above quantities,  the aggregated closed-loop system associated with the followers can be succinctly described  via

	\begin{equation}
	\begin{split}
	\label{eq:followerclose}
	x^f_{t+1} &=\underbrace{\left(\overline A_f+\overline B_fL_{ff}\mathcal \overline C_f\right)}_{ A_f} x^f_{t}+\underbrace{\overline B_fL_{lf}C_l}_{B_f}x^l_{t},\\  
	w^f_{t}   &=  C_fx^f_{t}.
	\end{split}
	\end{equation}

We also record the aggregated dynamics for the leaders as 
 
	\begin{equation}
	\begin{split}
	\label{eq:leaderagg}
	x^l_{t+1} &=A_l x^l_{t}+B_l u_{t},\\  
	w^l_t&= C_lx^l_t, 
	\end{split}
	\end{equation}
where
	\begin{equation}\label{eq:nodematricesfollower}
	\begin{split}
	A_l &:=\diag
	(A_{N_f+1}, \ldots, A_{N}),\\
	B_l  &:=\diag
	(B_{N_f+1}, \ldots, B_{N}),\\
	C_l  &:=\diag
	(C_{N_f+1}, \ldots, C_{N}),\\
	x^l_t&:=
	\begin{pmatrix}
	x^{N_f+1}_t\\\vdots\\x^{N}_t
	\end{pmatrix}\in \mathbb{R}^{n_l},\\
	w^l_t&:=
	\begin{pmatrix}
	w^{N_f+1}_t\\\vdots\\w^{N}_t
	\end{pmatrix}\in \mathbb{R}^{p_l},
		\end{split}
	\end{equation}
	with dimensions  $ n_l =\sum_{i=N_f+1}^{N}n_i$ and  $p_l=\sum_{i=N_f+1}^{N}p_i$.
 
 \begin{figure}[!t]
 	\begin{center}
 		\includegraphics[width=8cm,height=6.5cm]{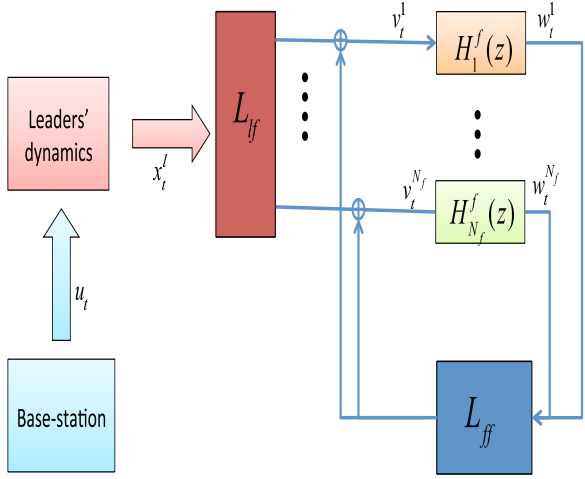}
 		\caption{\emph{ The dynamics of each follower are represented by $H^f_i(z)$, $i=1,\ldots,N_f$.}} \label{fig:structure2}
 	\end{center}
 \end{figure}
 
In this paper, our objective is to address the following question

\textbf{Under which  {\crr conditions can   states of    followers  be steered into}  any desired values from any intial conditions, using the command signal  $u_t$ and  control law \eqref{eq:followu}.}

{\crb
To this end, we first introduce  Fig. \ref{fig:structure2} that provides a detailed pictorial description of the framework under study. It is clear that indeed there exist two levels of control in this framework i.e. from the base-station to leaders and from the leaders to followers. 
}

\section{Reachability of Networked Systems} \label{sec:main}

We start this section by formally introducing  definitions of reachability  for each levels of control in  Fig. \ref{fig:structure2}. These definitions   are adapted from the literature  \cite{hespanha2009} for the purpose of the current paper.  

\begin{definition}\label{def:basereach}
	The follower dynamics \eqref{eq:followerclose} is said to be \textbf{leader-reachable} if and only if for any intial state $x^f_{t_0} \in \mathbb R ^{n_f}$ and an arbirary final state $\bar x^f_{t_f} \in \mathbb R ^{n_f}$, there exists $x^l_t$, $t \in [t_0,t_f]$ such that $x^f_{t_f}=\bar x^f_{t_f}$. 
\end{definition}

Similarly for the  dynamics \eqref{eq:leaderagg}, we state the following definition.

\begin{definition} \label{def:reachlead}
	The leader dynamics \eqref{eq:leaderagg} is said to be \textbf{base-reachable} if and only if for any intial state $x^l_{t_0} \in \mathbb R ^{n_l}$ and an arbirary final state $\bar x^l_{t_f} \in \mathbb R ^{n_l} $, there exists $u_t$, $t \in [t_0,t_f]$ such that $x^l_{t_f}=\bar x^l_{t_f}$. 
\end{definition}

The following lemma follows standard  systems and control literature see e.g.  \cite{kailath}. 

\begin{lemma}\label{lem:reachmain}
	The system   \eqref{eq:followerclose}/\eqref{eq:leaderagg}  is  leader-reachable/base-reachable if and only if the   matrix $\mathcal R_l=  \left(B_f,\;  A_f   B_f,\;\dots,\;  A^{n_f-1}_f  B_f\right) $/ $\mathcal R_b= \left( B_l,\;A_lB_l,\;\dots,\;A^{n_l-1}_lB_l \right)$  has full-row rank. 
\end{lemma}

The above definitions enable  us to introduce the following lemma. 

{\crr{\crb
\begin{lemma} \label{lem:folowerbasereach}
Suppose the system \eqref{eq:followerclose}  
is leader-reachable and  the system \eqref{eq:leaderagg} is  base-reachable. Then there exists   $u_t$, $t \in [t_0,t_f]$, such that 
for any arbirary final state $\bar x^f_{t_f}$, $x^f_{t_f}=\bar x^f_{t_f}$. 
\end{lemma}
}
\begin{proof}
Under the conditions in the lemma statement the reachable subspaces for the system  \eqref{eq:followerclose} and \eqref{eq:leaderagg} are  equal to $\mathbb R^{n_f}$  and $\mathbb R^{n_l}$ respectively. Thus, one can always construct a  proper input signal $u_t$  \cite{hespanha2009} to steer the  states of the system \eqref{eq:followerclose}  into any desired value. 
\end{proof}

%
%

}
{\crr
The result of Lemma \ref{lem:folowerbasereach}  is illustrated further in the following example. 

	\begin{figure}[!t]
		\begin{center}
			\includegraphics[width=4cm,height=4cm]{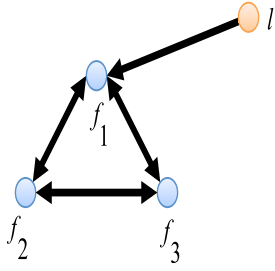}
			\caption{\emph{Followers are colored in blue and denoted by $f_i$ and the single leader is yellow indicated by $l$. }} \label{fig:3agents}
		\end{center}
	\end{figure}

\begin{example}

Consider a setup as  shown in Fig. \ref{fig:3agents}. { \crb For the sake of  illustration, we suppose that   all subsystems including followers and the leader  have very simple  dynamics  described as follows}	
\begin{equation}\label{eq:ex1}
\bar x^i_{t+1}=0.2 \bar x^i_t + \bar v^i_t, \:\:i=1,\ldots,4,
\end{equation}
with $\bar v^i_{t}=\sum_{j=1}^4 \bar L_{ij}\bar x^j_{t}$, $i=1,2,3$ and $\bar v^4_{t}=u_t$.

Given the dynamics (\ref{eq:ex1}), all subsystems are reachable. We let the parameters $\bar  L_{21}= \bar L_{31}=1$, $\bar L_{21}=\bar L_{23}=2$ and $\bar L_{31}=\bar L_{32}=3$. Then it is easy to verify that the follower dynamics are 

\begin{equation}
\begin{pmatrix}
\bar  x^1_{t+1}\\
\bar  x^2_{t+1}\\
\bar   x^3_{t+1}
\end{pmatrix}=
\underbrace{\begin{pmatrix}
0.2 & 1 & 1\\
2 & 0.2& 2\\
3 & 3 & 0.2
\end{pmatrix}}_{A_f}
\begin{pmatrix}
\bar x^1_{t}\\
\bar x^2_{t}\\
\bar x^3_{t}
\end{pmatrix}+
\underbrace{\begin{pmatrix}
1\\
0\\
0
\end{pmatrix}
}_{B_f}\bar x^4_{t+1}. \label{eq:exaggre}
\end{equation}

It can be checked that this system is base-reachable. Given the dynamics (\ref{eq:ex1}), one can conclude that  the states of  the system \eqref{eq:exaggre} can be driven into any desired point in the space using the input command $u_t$. 	
\end{example}

\subsection{Leader Reachability}

}

%
%

{\crb
In the previous subsection, we introduced the notions of leader-reachability and base-reachability.  It is   worthwhile to investigate these notions when  networked systems attain special interconnection structures. This is because in different applications,  subsystems may be linked to each other in particular forms see e.g. \cite{olfati2002distributed}, \cite{nabi2013controllability}, \cite{brockett-willems}. Thus, in this subsection, we aim to explore networked systems with special structures. 	
}

One should note that   when the pairs $A_i,B_i$ are reachable,  the base-reachability of the system \eqref{eq:leaderagg} becomes immediate. {\crr However, it still remains a nontrivial task to explore the concept of leader-reachability for the system \eqref{eq:followerclose}. In this subsection we study this notion in more detail.

The analysis of leader-reachability for the system \eqref{eq:followerclose}  is very intricate in general. This is because the state matrix $A_f$ has an involved structure. Furthermore, networks with special coupling structures appear in many
applications, such as cyclic pursuit \cite{MarshallBrouckeFrancis};
shortening flows in image processing \cite{Bruckstein} or the
discretization of partial differential
equations~\cite{brockett-willems}. Thus,  in order to provide some rigorous results we study the notion of leader-reachability  when the state matrix attains  some particular structures.     Here,  we consider two scenarios namely  symmetric $A_f$   and circulant  $A_f$.     
 	
\subsubsection{Symmetric $A_f$}	
  Several interconnection topologies can lead to a symmetric $A_f$ matrix. For instance, consider a scenario where a set of  scalar subsystems are connected to each other based on the consensus law \cite{olfati2002distributed}.
}

\begin{theorem} \label{them:symmetricA}
	Suppose that the matrix $A_f$ is symmetric. Let $\lambda_i$ and $\nu_i$, $\forall i \in \{1,2,\ldots,n_f\}$,  denote eigenvalues and the corresponding eigenvectors of $A_f$ and $B_f=\left( b_f^{1},\;\ldots,\; b_f^{m_f}\right)$. Then the dynamics \eqref{eq:followerclose} is leader-reachable if $\lambda_i\ne\lambda_j$ and $\nu^T_ib_f^j\ne0$ $\forall i,j$.  
\end{theorem}

\begin{proof} 
	In this case, the matrix $A_f$ can be written as $Q\Lambda Q^{\top}$ where $Q$ is an orthonormal matrix comprised of $\nu_i$  and $\Lambda$ is a diagonal matrix  containing  eigenvalues of $A_f$. It is easy to see that 
	\[
	\begin{split}
	\mathcal{R}_l&=\left(B_f,
	\;Q\Lambda Q^{\top} B_f,\ldots,Q\Lambda^{n_f-1} Q^{\top}B_f\right)\\
	&= Q \underbrace{\left(Q^{\top}B_f,\; \Lambda Q^{\top} B_f,\ldots, \Lambda^{n_f-1} Q^{\top}B_f\right).}_{\overline{\mathcal R}_l}\
	\end{split}
	\]
	The matrix $Q$ has full rank. Thus, the rank of $\mathcal R_l$ is  determined by $\overline{\mathcal{R}}_l$  that is expressed as 
	
	\[
	\begin{split}
	&\left(
	\begin{pmatrix}
	\nu_1^\top \\ \vdots \\ \nu_{n_f}^\top 
	\end{pmatrix} \begin{pmatrix}
	b_f^1&\ldots&b_f^{m_f}
	\end{pmatrix},\; \begin{pmatrix}
	\lambda_1& &\\
	&\ddots&\\
	& &\lambda_{n_f}
	\end{pmatrix} \begin{pmatrix}
	\nu_1^\top\\ \vdots\\  \nu_{n_f}^\top 
	\end{pmatrix} \right.\\  	&\left.\begin{pmatrix}
	b_f^1&\ldots&b_f^{m_f}
	\end{pmatrix}, \; \ldots,
	\begin{pmatrix}
	\lambda_1& &\\
	&\ddots&\\
	& &\lambda_{n_f}
	\end{pmatrix} ^{n_f-1} \begin{pmatrix}
	\nu_1^\top\\ \vdots\\  \nu_{n_f}^\top 
	\end{pmatrix}\right.\\ &\left. \begin{pmatrix}
	b_f^1&\ldots&b_f^{m_f}
	\end{pmatrix} 
	\right.\Bigg).	\end{split}
	\]

	By appealing to the theorem assumptions and the fact that $\nu_i^\top\nu_j\ne0 $ $\forall i\ne j$, the result immediately follows.  
\end{proof}
{\crr
\subsubsection {Circulant $A_f$}
In this subsection,  we study the situation where the matrix $A_f$ has circulant structure.  This situation may happen  naturally when the interconnection topology is a circulant graph see e.g \cite{nabi2013controllability}.  It is worthwhile noting that networked systems with circulant topology arise in different applications such as quantum communication \cite{Ilic2011} and complex memory management \cite{Wong1974}.

The following example illustrates a situation where the matrix $A_f$ acquires a circulant structure.

}
\begin{example}
	Let us consider a network consisting of four identical  single-output-single-output (SISO) subsystems. We suppose the  dynamics for each subsystem  are expressed as  {\crr
	\[
	\begin{split}
	\label{sysexample}
	\hat x^i_{t+1} &=a \hat x^i_{t}+b \hat v^i_{t},\\  
	\hat w^i_{t}   &=\hat x^i_{t},\; i=1,\dots,4.
	\end{split}
	\]
	with $|a|<1$. $\hat  v^i_{t}=\sum_{j=1}^4 \hat  L_{ij}\hat x^j_{t}$, $i=1,2,3$ and $\hat v^4_{t}=u_t$.  
	
	As shown in Fig. \ref{fig:agentstructure}, the interconnection parameters i.e. $\hat L_{ij}$ are set as $\hat L_{12}=\hat L_{23}=\alpha_1$, $\hat L_{13}=\hat L_{21}= \hat L_{32}=\alpha_2$, $\hat L_{32}=\alpha_3$, $\hat L_{14}=b$ and $\hat L_{24}=\hat L_{34}=0 $. }  Then it is easy to verify that $A_f=\begin{pmatrix}a&\alpha_1&\alpha_2\\\alpha_2&a&\alpha_1\\ \alpha_3&\alpha_2&a\end{pmatrix}$ and $B_f=\begin{pmatrix}
	b\\0\\0
	\end{pmatrix}$.  {\crr We set parameter of  dynamics and topology  to be  $a=0.2$, $b=\alpha_1=\alpha_3=1$  and $\alpha_2=0.5$. Then it can be checked that  the whole network depicted in Fig. \ref{fig:agentstructure} is  reachable.}

	\begin{figure}[!t]
		\begin{center}
			\includegraphics[width=5cm,height=5cm]{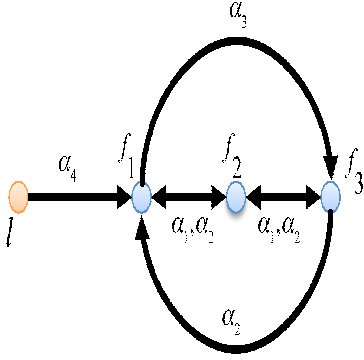}
			\caption{\emph{Followers are colored in blue and denoted by $f_i$ and the sole leader is yellow indicated by $l$. The weighting coefficients on connecting links are represented by $\alpha_i$.  }} \label{fig:agentstructure}
		\end{center}
	\end{figure}

\end{example}

As mentioned earlier, the matrix $A_f$ in the above example has a particular form known as \textit {circulant}. Thus,  we now investigate  in more detail a scenario where  the matrix $A_f$   has circulant structure  i.e. is of the form 
\begin{equation*}
\begin{split}
A_f &=\rm{Circ} (\alpha_0,...,\alpha_{n_f-1}) \\&=\begin{pmatrix}
\alpha_0 & \alpha_1 & \cdots &\alpha_{n_f-2}& \alpha_{n_f-1} \\
\alpha_{n_f-1} & \alpha_0 &\alpha_1& \cdots & \alpha_{n_f-2} \\
\vdots&\ddots&\ddots& \ddots&\vdots\\
\alpha_2 & \cdots & \alpha_{n_f-1} & \alpha_0& \alpha_{1} \\
\alpha_1 & \alpha_2&\cdots & \alpha_{n_f-1}& \alpha_0
\end{pmatrix}.
\end{split}
\end{equation*}

It is well-known that circulant matrices \cite{davis} are  diagonalizable by the \textbf{Fourier matrix}
\begin{equation*}
\begin{split}
\Phi&=\frac{1}{\sqrt{n_f}}\begin{pmatrix} 1&1&1& \ldots & 1\\
1& \omega &\omega^2 &\ldots & \omega^{n_f-1} \\
1 & \omega^2 & \omega^4 & \ldots & \omega^{2n_f-2}\\
\vdots & & & & \\
1 & \omega^{n_f-1} & \omega^{2n_f-2} & \ldots & \omega^{(n_f-1)^2}
\end{pmatrix} \text{ ,}\\
&=\frac{1}{\sqrt{n_f}} \begin{pmatrix}
\phi_1\\ \phi_2\\ \vdots \\ \phi_{n_f} 
\end{pmatrix}
\end{split}
\end{equation*}
where $\omega=e^{2 \pi j /{n_f}}$  denotes a primitive $n_f-$th root of
unit and $\phi_i$ denote rows of $\Phi$. Note, that $\Phi$ is both a unitary and a symmetric
matrix.
It is then easily seen that any circulant matrix
$L$ has the form
$
A_f=\Phi \diag(p_L(1),p_L(\omega),\ldots,p_L(\omega^{n_f-1}))\Phi^*,
$ $= \Phi \Gamma \Phi^*$
where
$
p_L(z):=\sum_{k=0}^{n_f-1}c_kz^{k-1}.
$
As a consequence of the preceding analysis we obtain the following result.

\begin{theorem} \label{them:circulantA}
	Suppose that the matrix $A_f$ is   circulant and $ B_f=\left( b_f^{1},\;\ldots,\; b_f^{m_f}\right)$. Then the dynamics \eqref{eq:followerclose} is leader-reachable if $ \phi^\top_ib_f^j\ne0$,  $\forall i,j$.  
\end{theorem}

\begin{proof}
	From the above analysis, one can write 
	\[
	\begin{split}
	\mathcal{R}_l&=\left(B_f,
	\;\Phi\Gamma \Phi^* B_f,\ldots,\Phi\Gamma^{n_f-1} \Phi^*B_f\right)\\
	&= \Phi \underbrace{\left(\Phi^*B_f,\; \Gamma \Phi^* B_f,\ldots, \Gamma^{n_f-1} \Phi^*B_f\right)}_{\underline{\mathcal R}_l}\
	\end{split}
	\]
	Now by using the same argument as in the proof of Theorem  \ref{them:symmetricA} 
	the result immediately follows. 
\end{proof}

\subsection{Generic Reachability}
{\crb  The previous subsection examined the leader-reachability and base-reachability notions for special network structures.	In this subsection,  we show that these properties hold in almost all cases.  To this end, we first need to define  the parameter space $\Theta$ as 
	
	\begin{equation}
	\begin{split}
	\Theta=&  \{ \mvec \left (A_1,\ldots, A_N\right), \mvec \left(B_1,\ldots,B_N\right),\\ 
	&\mvec \left(C_1,\ldots,C_N\right), \mvec(L)\}.
	\end{split}
	\end{equation}
	Then we recall the notion of generic set from \cite{anderson2015}. A subset of the parameter space $\Theta$ is said to be generic if it is
	an open and dense in $\Theta$. 
} We now use this notion to introduce the next results.

\begin{theorem}
	The systems \eqref{eq:followerclose} and \eqref{eq:leaderagg} are leader-reachable and base-reachable  on a generic subset of the parameter space $\Theta$. 
\end{theorem}

\begin{proof}
	First,  one can easily find a set of matrices $A_i$,$B_i$, etc., such that  the associated matrix $\mathcal R_l$ attains full- row rank.  Second, let $\sigma_i$  $i=1,\ldots,n_fm_f$ denote the  columns of $\mathcal R_l$ defined in Lemma \ref{lem:reachmain}. Then note that  the system \eqref{eq:followerclose} is not reachable  if and only if 
	
	\begin{equation} \label{eq:zeros}
	\det\{\Gamma\}=0,
	\end{equation}
	where $\Gamma \in \mathbb R^{n_f \times n_f}$ and the columns of $\Gamma$ are constructed by selecting  any $n_f$ choice of $\sigma_i$. Then the set of zeros of  \eqref{eq:zeros} defines a proper algebraic set. Therefore, its complement, which is associated with all reachable systems, is the complement of a proper algebraic set and hence is open and dense in the parameter space. The latter is equivalent to the statement of the theorem. Finally, note that those parts of the theorem statement asssociated with the system \eqref{eq:leaderagg} become  trivial in the light of  \cite{wonhem1979} pages 44-45.  
\end{proof}

The preceding  result roughly suggests that for almost all choices of parameter matrices $A_i$, $B_i$ and etc.,  there exists a $u_t$ that can steer the follower and leader states to  desired values.

\section{Conclusion and Future Works}\label{sec:con}
We examined the reachability problem for   networked systems.  It was assumed that all subsystems are expressed by discrete linear time-invariant state space models.

We considered the network topology to be time-invariant.  We addressed a hierarchical framework where there exists  a base-station at the highest level; superior subsystems (leaders) are at an  intermediate level  and the rest of subsystems (followers) stay at the final stage.   The followers are only able to communicate with each other   and with leaders only.  We introduced notions of   leader-reachability and base-reachability.  We explored situations under which the algebraic criteria associated with these notions are satisfied.  It turned out that the reachability of  leaders is  enough for   fulfilling base-reachability. We then studied  leader-reachability and provided algebraic conditions for this property to hold. We examined different topologies such as those that give  rise to symmetric and circulant state matrices. We further demonstrated that when the system parameters assume generic values, the whole network is reachable.

There are several interesting problems that still remain open. The scenarios discussed in this paper only cover    certain classes of linear networked systems. It would be  of interest to provide a result that  includes   more general instances.     Another problem involves   studying  reachability for a scenario where   interconnection matrices  assume only  zero and free entries.  This problem is highly related with the structural controllability problem studied in the literature \cite{lin1974structural}. Another interesting issue is associated with control energy of the whole networked system. In particular, we are interested in  designing  topologies such that reachability is preserved but the deployed control energy remains within  some given boundaries as well.

\section*{Acknowledgments}
The support by the Australian Research Council (ARC) is gratefully acknowledged.

\bibliographystyle{elsarticle-num}
\bibliography{thesis}

\begin{thebibliography}{10}
\expandafter\ifx\csname url\endcsname\relax
  \def\url#1{\texttt{#1}}\fi
\expandafter\ifx\csname urlprefix\endcsname\relax\def\urlprefix{URL }\fi
\expandafter\ifx\csname href\endcsname\relax
  \def\href#1#2{#2} \def\path#1{#1}\fi

\bibitem{sinopoli2003distributed}
B.~Sinopoli, C.~Sharp, L.~Schenato, S.~Schaffert, S.~Sastry, Distributed
  control applications within sensor networks, Proceedings of the IEEE 91~(8)
  (2003) 1235--1246.

\bibitem{olfati2002distributed}
R.~Olfati-Saber, R.~M. Murray, Distributed cooperative control of multiple
  vehicle formations using structural potential functions, in: Proc the IFAC
  World Congress, 2002, pp. 346--352.

\bibitem{tanner2003stable}
H.~G. Tanner, A.~Jadbabaie, G.~J. Pappas, Stable flocking of mobile agents part
  {I}: dynamic topology, in: Proc the IEEE Conference on Decision and Control,
  2003, pp. 2016--2021.

\bibitem{Olfati2007}
R.~Olfati-Saber, J.~A. Fax, R.~M. Murray, Consensus and cooperation in
  networked multi-agent systems, Proceedings of the IEEE 95~(1) (2007)
  215--233.

\bibitem{dankers2014system}
A.~G. Dankers, System identification in dynamic networks, Ph.D. thesis, TU
  Delft, Delft University of Technology (2014).

\bibitem{hespanha2007survey}
P.~J. Hespanha, P.~Naghshtabrizi, Y.~Xu, A survey of recent results in
  networked control systems, IEEE Proceedings 95~(1) (2007) 138--162.

\bibitem{zamaniautomatic2014}
M.~Zamani, U.~Helmke, B.~D.~O. Anderson, Zeros of networked systems with
  time-invariant interconnections,{ submitted for publication
  (http://arxiv.org/abs/1408.6889)}.

\bibitem{kailath}
T.~Kailath, Linear Systems, Prentice-Hall, New Jersey, 1980.

\bibitem{tanner2004}
H.~G. Tanner, On the controllability of nearest neighbor interconnections, in:
  Proc the IEEE Conference on Decision and Control, Vol.~3, 2004, pp.
  2467--2472.

\bibitem{ji2008controllability}
Z.~Ji, H.~Lin, T.~H. Lee, Controllability of multi-agent systems with switching
  topology., in: Proc of the IEEE conference on Robotics, Automation and
  Mechatronics, 2008, pp. 421--426.

\bibitem{ji2007graph}
M.~Ji, M.~Egerstedt, A graph-theoretic characterization of controllability for
  multi-agent systems, in: Proc the American Control Conference, 2007, pp.
  4588--4593.

\bibitem{rahmani2009controllability}
A.~Rahmani, M.~Ji, M.~Mesbahi, M.~Egerstedt, Controllability of multi-agent
  systems from a graph-theoretic perspective, SIAM Journal on Control and
  Optimization 48~(1) (2009) 162--186.

\bibitem{liu2011controllability}
Y.-Y. Liu, J.-J. Slotine, A.-L. Barab{\'a}si, Controllability of complex
  networks, Nature 473~(7346) (2011) 167--173.

\bibitem{zamani2009structural}
M.~Zamani, H.~Lin, Structural controllability of multi-agent systems, in: Proc
  the American Control Conference, 2009, pp. 5743--5748.

\bibitem{martini2010controllability}
S.~Martini, M.~Egerstedt, A.~Bicchi, Controllability analysis of multi-agent
  systems using relaxed equitable partitions, International Journal of Systems,
  Control and Communications 2~(1) (2010) 100--121.

\bibitem{Ji2006}
M.~Ji, A.~Muhammad, M.~Egerstedt, Leader-based multi-agent coordination:
  controllability and optimal control, in: Proc the American Control
  Conference, 2006.

\bibitem{rahmani2006}
A.~Rahmani, M.~Mesbahi, On the controlled agreement problem, in: Proc the
  American Control Conference, 2006, 2006.

\bibitem{ji2008graph}
Z.~Ji, H.~Lin, T.~H. Lee, A graph theory based characterization of
  controllability for multi-agent systems with fixed topology, in: Proc the
  IEEE Conference on Decision and Control, 2008, pp. 5262--5267.

\bibitem{partovi}
A.~Partovi, L.~Hai, J.~Zhijian, Structural controllability of high order
  dynamic multi-agent systems, in: Proc the IEEE Conference on Robotics,
  Automation and Mechatronics, 2010.

\bibitem{Ji2009}
L.~Wang, F.~Jiang, G.~Xie, Z.~Ji, Controllability of multi-agent systems based
  on agreement protocols, Science in China Series F: Information Sciences
  52~(11) (2009) 2074--2088.

\bibitem{fax2004}
J.~A. Fax, R.~M. Murray, Information flow and cooperative control of vehicle
  formations, IEEE Transactions on Automatic Control 49~(9) (2004) 1465--1476.

\bibitem{hespanha2009}
J.~P. Hespanha, Linear Systems Theory, Princton University Press, 2009.

\bibitem{nabi2013controllability}
M.~Nabi-Abdolyousefi, M.~Mesbahi, On the controllability properties of
  circulant networks, IEEE Transactions on Automatic Control 58~(12) (2013)
  3179--3184.

\bibitem{brockett-willems}
R.~W. Brockett, J.~L. Willems, Discretized partial differential equations:
  Examples of control systems defined on modules, Automatica 10~(5) (1974)
  507--515.

\bibitem{MarshallBrouckeFrancis}
J.~A. Marshall, M.~E. Broucke, B.~A. Francis, Formations of vehicles in cyclic
  pursuit, IEEE Transactions on Automatic Control 49~(11) (2004) 1963--1974.

\bibitem{Bruckstein}
M.~A. Bruckstein, G.~Sapiro, D.~Shaked, Evolutions of planar polygons.,
  International Journal of Pattern Recognition and Artificial Intelligence
  9~(6) (1995) 991--1014.

\bibitem{Ilic2011}
A.~Ili\'{c}, M.~Ba\v{s}i\'{c}, New results on the energy of integral circulant
  graphs, Applied Mathematics and Computation 218~(7) (2011) 3470 -- 3482.

\bibitem{Wong1974}
C.~K. Wong, D.~Coppersmith, A combinatorial problem relatedto multimodule
  memory organizations, Journal of the ACM 21~(3) (1974) 392--402.

\bibitem{davis}
P.~J. Davis, Circulant matrices, John Wiley and Sons. New York, 1979.

\bibitem{anderson2015}
B.~D. Anderson, M.~Deistler, E.~Felsenstein, B.~Funovits, L.~Koelbl, M.~Zamani,
  Multivariate {AR} systems and mixed frequency data: G-identifiability and
  estimation, Econometric Theory (2015) 1--34.

\bibitem{wonhem1979}
W.~M. Wonham, Linear Multivariable Control: a Geometric Approach,
  Springer-Verlag, New York, 1979.

\bibitem{lin1974structural}
C.-T. Lin, Structural controllability, IEEE Transactions on Automatic Control
  19~(3) (1974) 201--208.

\end{thebibliography}

\end{document}